\documentclass[12pt,a4paper]{article}

\usepackage{amsfonts}
\usepackage{graphicx}
\usepackage{amsmath}
\usepackage{amsthm}
\usepackage{longtable}
\usepackage{float}
\usepackage{cite}
\usepackage{adjustbox}
\usepackage{tikz}
\usepackage{caption}

\theoremstyle{definition}
\newtheorem{definition}{Definition}
\newtheorem{theorem}{Theorem}
\newtheorem{proposition}{Proposition}

\usetikzlibrary{shapes.geometric,arrows}
\tikzstyle{startstop}=[rectangle,rounded corners,minimum width=3cm,minimum height=1cm,align=left,draw=black]
\tikzstyle{io}=[trapezium,trapezium left angle=70,trapezium right angle=110,minimum width=3cm,minimum height=1cm,align=left,draw=black]
\tikzstyle{process}=[rectangle,minimum width=3cm,minimum height=1cm,align=left,draw=black]
\tikzstyle{decision}=[diamond,minimum width=3cm,minimum height=1cm,align=left,draw=black]
\tikzstyle{arrow}=[thick,->,>=stealth]

\begin{document}

\title{Special core tensors of multi-qubit states and the concurrency of three lines}

\author{Choong Pak Shen\footnote{pakshenchoong@gmail.com} \and Hishamuddin Zainuddin \and Chan Kar Tim\footnote{Corresponding author: chankt@upm.edu.my} \and Sh. K. Said Husain}
\date{Institute for Mathematical Research, Universiti Putra Malaysia, 43400 Serdang, Selangor, Malaysia.\\
Date: \today}

\maketitle

\begin{abstract}
Classification of multipartite states aims to obtain a set of operationally useful and finite entanglement classes under the action of either local unitary (LU) or stochastic local operation and classical communication (SLOCC). In this work, we propose a computationally simple approach to find these classes by using higher order singular value decomposition (HOSVD) and the concurrency of three lines. Since HOSVD simultaneously diagonalizes the one-body reduced density matrices (RDM) of multipartite states, the core tensor of multipartite states is the pure-state representation of such simultaneously diagonalized one-body RDM. We identified the special core tensors of three and four qubits, which are also genuinely entangled by default. The special core tensors are further categorized into families of states based on their first $n$-mode singular values, $\sigma_1^{(i)2}$. The current proposal is limited to multi-qubit system, but it scales well with large multi-qubit systems and produces a finite number of families of states.
\end{abstract}

\section{Introduction}

Being a quantum resource under the local operation and classical communication (LOCC) paradigm \cite{Wootters1998,Chitambar2019}, numerous efforts have been dedicated to understand entanglement from various perspectives and mathematical tools \cite{Acin2000,Carteret2000,Kus2001,Sudbery2001,Sinolecka2002,Verstraete2002,Vidal2002,Lamata2006,Lamata2007,Plenio2007,Li2009,Chitambar2010,Kraus2010A,Kraus2010B,Sharma2010,Sawicki2011A,Sawicki2011B,Liu2012,Sawicki2012,Sharma2012,Li2013,Sharma2013,Walter2013,Li2014,Schwaiger2015,Sawicki2018}. To date, even though there is no single unified approach to describe multipartite entanglement, most discussions focus around the operational aspects of entanglement in quantum information processing tasks. Since the local unitary (LU) or stochastic local operation and classical communication (SLOCC) entanglement classes of multipartite states are claimed to be infinite \cite{Dur2000,Gharahi2018}, the current challenge in the classification of multipartite states is to find a computationally simple approach that gives operationally meaningful and finite classification results \cite{Gharani2020,Gharani2021}.

Previously \cite{Choong2020}, we showed that higher order singular value decomposition (HOSVD) \cite{Lathauwer2000,Kolda2009} simultaneously diagonalizes the one-body reduced density matrices of three qubits. Furthermore, by finding all the solutions to the all-orthogonality conditions of three qubits, we recovered all the special states of three qubits \cite{Carteret2000}. The first $n$-mode singular values, $\sigma_1^{(n)2}$, where $n=1,\,2,\,3$, can be used to plot a LU entanglement polytope similar to that in \cite{Walter2013}. However, as the number of variables grows exponentially with the increase in the number of subsystems, solving the all-orthogonality conditions of multipartite states is not a feasible approach in generalizing the methodology to multipartite systems.

Before we proceed further, we would like to point out that from our previous results, some special states of three qubits are specific cases to a more generic setting. For example, the bi-separable states $C|AB$ with the first $n$-mode singular values $(\sigma_1^{(1)2},\,\sigma_1^{(2)2},\,\sigma_1^{(3)2}) = (\sigma_1^{(1)2},\,\sigma_1^{(1)2},\,1)$
\begin{align*}
\left| \text{Bi-Sep}_{C|AB} \right\rangle = t_{111} \left| 111 \right\rangle + t_{221} \left| 221 \right\rangle
\end{align*}
and the three-qubit states with $(\sigma_1^{(1)2},\,\sigma_1^{(2)2},\,\sigma_1^{(3)2}) = (\frac{1}{2},\,\frac{1}{2},\,\sigma_1^{(3)2})$ are specific cases to the following Slice states,
\begin{align*}
& \left| \text{S}_1 \right\rangle = t_{111} \left| 111 \right\rangle + t_{112} \left| 112 \right\rangle + t_{221} \left| 221 \right\rangle + t_{222} \left| 222 \right\rangle, \\
& \bar{t}_{111} t_{112} + \bar{t}_{221} t_{222} = 0,
\end{align*}
where $\sigma_1^{(3)2} > \sigma_1^{(1)2} = \sigma_1^{(2)2}$. Therefore, our current work focuses on identifying the generic special states of a multipartite system since the specific cases are inclusive to the generic special states that we identified.

In this work, we propose a computationally simple approach to identify the special states of multi-qubit core tensors. This is important because core tensors are also the pure-state representation of multi-qubit states when their one-body reduced density matrices (RDM) are simultaneously diagonalized. Based on the concurrency of three lines \cite{Pedoe1970}, we convert the problem of finding solutions to the set of all-orthogonality conditions into the problem of satisfying a set of determinants to be zero. This conversion has the added computational advantage in that satisfying the requirements for a set of determinants to be zero is easier than finding the solutions to a set of polynomial equations. To do so, we define a pair of conjugate concurrent variables (CCV) so that the one-to-one correspondence between the algebraic manipulations of a set of simultaneous equations and the geometrical idea based on the concurrency of three lines is preserved. Then, we describe a general algorithm of this approach and demonstrate it with the case of four qubits. Even though our approach is unable to identify the generalized GHZ states, these states have a very recognizable form.

We structure our paper as follows. In Section \ref{SecMatrixUnfoldingAndHOSVD}, we provide the original definitions of matrix unfolding and HOSVD. We show that matrix unfolding is related to the RDM of multipartite states and HOSVD simultaneously diagonalizes the one-body RDM of multipartite states. In Section \ref{SecConcurrencyThreeLines}, we summarize our previous results on three qubits, and show that the derivation from our previous work is equivalent to a geometrical concept in projective geometry, called the concurrency of three lines. By solving the concurrency of three lines for three qubits, we identify all the special three-qubit core tensors using this new approach. Finally, we state a general algorithm for this approach on multi-qubit core tensors in Section \ref{SecFourQubitsAndGeneralization}, and demonstrate it with the case of four qubits.

\section{Matrix unfolding and higher order singular value decomposition} \label{SecMatrixUnfoldingAndHOSVD}

\subsection{Matrix unfolding}

The Hilbert space of a composite quantum system is given by the tensor product of its subsystems' Hilbert spaces. Because of this, the probability amplitudes of multipartite states are elements of higher order tensors, allowing us to make use of tensor decomposition in the classification of multipartite states \cite{Liu2012,Li2013,Li2014,Choong2020}. In order to write down higher order tensors in a way that obeys the matrix-tensor multiplication rules, a formalism called matrix unfolding \cite{Lathauwer2000} or matricization \cite{Kolda2009} of tensors was previously introduced.

\begin{definition}[Matrix unfolding \cite{Lathauwer2000}] \label{MatrixUnfolding}
Let $\Psi \in \mathbb{C}^{I_1} \otimes \ldots \otimes \mathbb{C}^{I_n} \otimes \ldots \otimes \mathbb{C}^{I_N}$ be an $N$th-order complex tensor. The $n$-th matrix unfolding, $\Psi_{(n)}$, is a matrix of size $I_n \times (I_{n+1} \times I_{n+2} \times \ldots \times I_N \times I_1 \times I_2 \times \ldots \times I_{n-1})$, whereby the tensor element $\psi_{i_1 i_2 \ldots i_n \ldots i_N}$ will be at the position with row index $i_n$ and column index
\begin{align}
& (i_{n+1} - 1) I_{n+2} I_{n+3} \ldots I_N I_1 I_2 \ldots I_{n-1} + (i_{n+2} -1) I_{n+3} I_{n+4} \ldots I_N I_1 I_2 \ldots I_{n-1} \nonumber \\ & + \ldots + (i_N -1) I_1 I_2 \ldots I_{n-1} + (i_1 -1) I_2 I_3 \ldots I_{n-1} + (i_2 -1) I_3 I_4 \ldots I_{n-1} \nonumber \\ & + \ldots + i_{n-1}.
\end{align}
\end{definition}

We redefine matrix unfolding by making use of the bra-ket notation.

\begin{definition}[Matrix unfolding in bra-ket notation] \label{MatrixUnfoldingAlternative}
Let $\Psi \in \mathbb{C}^{I_1} \otimes \ldots \otimes \mathbb{C}^{I_n} \otimes \ldots \otimes \mathbb{C}^{I_N}$ be an $N$th-order complex tensor. In the bra-ket notation, the $n$-th matrix unfolding, $\Psi_{(n)}$, rewrites $\Psi$ into the following matrix form,
\begin{align*}
\Psi_{(n)} = \sum_{i_1 \ldots i_N} \psi_{i_1 \ldots i_N} \left| i_n \right\rangle \left\langle i_{n+1} \ldots i_N i_1 \ldots i_{n-1} \right|.
\end{align*}
\end{definition}

From Definition \ref{MatrixUnfoldingAlternative}, we propose the following. The proof can be found in Appendix \ref{AppendixProofMatrixUnfolding}.

\begin{proposition}[Matrix unfolding and reduced density matrices] \label{MatrixUnfoldingRDMProposition}
The $n$-th matrix unfolding $\Psi_{(n)}$ of an $N$-th order tensor $\Psi \in \mathbb{C}^{I_1} \otimes \ldots \otimes \mathbb{C}^{I_n} \otimes \ldots \otimes \mathbb{C}^{I_N}$ is related to its one-body and $(n-1)$-body reduced density matrices, $\rho_n$ and $\rho_{n+1\, \ldots N\, 1\, \ldots n-1}$ respectively, through the following relations,
\begin{align}
\Psi_{(n)} \Psi_{(n)}^\dagger & = \rho_n, \label{MatrixUnfoldingOneBodyRDM} \\
\Psi_{(n)}^\text{T} \bar{\Psi}_{(n)} & = \rho_{n+1\, \ldots N\, 1\, \ldots n-1}. \label{MatrixUnfoldingMinusOneRDM}
\end{align}
\end{proposition}

\subsection{Higher order singular value decomposition}

Next, we introduce higher order singular value decomposition (HOSVD) \cite{Lathauwer2000} and its matrix unfolding variant \cite{Lathauwer2000,Li2013}.

\begin{theorem}[Higher order singular value decomposition \cite{Lathauwer2000}] \label{HOSVDTheorem}
Let $\Psi \in \mathbb{C}^{I_1} \otimes \ldots \otimes \mathbb{C}^{I_n} \otimes \ldots \otimes \mathbb{C}^{I_N}$ be an $N$th-order complex tensor. There exists a core tensor $\mathcal{T}$ of $\Psi$ and a set of unitary matrices $U^{(1)},\, \ldots,\, U^{(n)},\, \ldots,\, U^{(N)}$ such that
\begin{align}
\Psi = U^{(1)} \otimes U^{(2)} \otimes \ldots \otimes U^{(n)} \otimes \ldots \otimes U^{(N)} \mathcal{T}. \label{HOSVD}
\end{align}
The core tensor $\mathcal{T}$ is also an $N$th-order complex tensor of which the subtensors $\mathcal{T}_{i_{n}=\alpha}$, obtained by fixing the $n$-th index to $\alpha$, have the properties of
\begin{enumerate}
\item{
\emph{All-orthogonality}: Two subtensors $\mathcal{T}_{i_{n}=\alpha}$ and $\mathcal{T}_{i_{n}=\beta}$ are orthogonal for all possible values of $n$, $\alpha$ and $\beta$, subject to $\alpha \neq \beta$:
\begin{align}
\left \langle \mathcal{T}_{i_{n}=\alpha}, \mathcal{T}_{i_{n}=\beta} \right \rangle & = \sum_{i_1 i_2 \ldots i_{n-1} i_{n+1} \ldots i_N} \bar{t}_{i_1 i_2 \ldots i_{n-1} \alpha i_{n+1} \ldots i_N} t_{i_1 i_2 \ldots i_{n-1} \beta i_{n+1} \ldots i_N} \nonumber \\ & = 0 \; \; \text{when} \; \; \alpha \neq \beta;
\end{align}
}
\item{
\emph{Ordering}:
\begin{align}
\left| \mathcal{T}_{i_{n}=1} \right| \geq \left| \mathcal{T}_{i_{n}=2} \right| \geq \ldots \geq \left| \mathcal{T}_{i_{n}=I_n} \right| \geq 0
\end{align}
for all possible values of $n$,
}
\end{enumerate}
where $t_{i_1 i_2 \ldots i_N}$ is the element of the tensor $\mathcal{T}$. The Frobenius norm of the subtensors $\left| \mathcal{T}_{i_n=i} \right|$ is given as
\begin{align}
\left| \mathcal{T}_{i_n=i} \right| & = \sqrt{\langle \mathcal{T}_{i_n=i},\, \mathcal{T}_{i_n=i} \rangle} \nonumber \\ & = \sqrt{\sum_{i_1 = 1}^{I_1} \ldots \sum_{i_{n-1} = 1}^{I_{n-1}} \sum_{i_{n+1} = 1}^{I_{n+1}} \ldots \sum_{i_N = 1}^{I_N} \bar{t}_{i_1 \ldots i_{n-1} i i_{n+1} \ldots i_N} t_{i_1 \ldots i_{n-1} i i_{n+1} \ldots i_N}} \nonumber \\ & = \sqrt{\sum_{i_1 = 1}^{I_1} \ldots \sum_{i_{n-1} = 1}^{I_{n-1}} \sum_{i_{n+1} = 1}^{I_{n+1}} \ldots \sum_{i_N = 1}^{I_N} \left| t_{i_1 \ldots i_{n-1} i i_{n+1} \ldots i_N} \right|^2}.
\end{align} 
and is called the $n$-mode singular value of $\Psi$, $\sigma_i^{(n)}$.
\end{theorem}

\begin{theorem}[Matrix unfolding of HOSVD \cite{Lathauwer2000,Li2013}] \label{HOSVDMatrixUnfoldingTheorem}
Let $\Psi \in \mathbb{C}^{I_1} \otimes \ldots \otimes \mathbb{C}^{I_n} \otimes \ldots \otimes \mathbb{C}^{I_N}$ be an $N$th-order complex tensor and $\mathcal{T}$ be its core tensor. The matrix unfolding of $\Psi$ and $\mathcal{T}$ can be obtained as
\begin{align}
\Psi_{(n)} = U^{(n)} T_{(n)} (U^{(n+1)} \otimes U^{(n+2)} \otimes \ldots \otimes U^{(N)} \otimes U^{(1)} \otimes U^{(2)} \otimes \ldots \otimes U^{(n-1)})^T, \label{HOSVDEquationMatrixUnfolding}
\end{align}
where $\Psi_{(n)}$ and $T_{(n)}$ are complex matrices of size $I_n \times (I_{n+1} \times I_{n+2} \times \ldots \times I_N \times I_1 \times I_2 \times \ldots \times I_{n-1})$, and $U^{(n)}$ are unitary matrices of size $I_n \times I_n$.
\end{theorem}

Due to Proposition \ref{MatrixUnfoldingRDMProposition}, we state the following. The proof can be found in Appendix \ref{AppendixProofHOSVD}.

\begin{theorem}[HOSVD and one-body reduced density matrices] \label{HOSVDReducedDensityMatrixTheorem}
Let $\Psi \in \mathbb{C}^{I_1} \otimes \ldots \otimes \mathbb{C}^{I_n} \otimes \ldots \otimes \mathbb{C}^{I_N}$ be an $N$th-order complex tensor and $\mathcal{T}$ be its core tensor. HOSVD simultaneously diagonalizes the set of one-body reduced density matrices of multipartite states in such a way that the $n$-mode singular values are ordered. The all-orthogonality conditions are the off-diagonal terms of the set of one-body reduced density matrices.
\end{theorem}

\section{Concurrency of three lines and three qubits} \label{SecConcurrencyThreeLines}

\subsection{Classification of three qubits} \label{SubSecClassificationThreeQubits}

In this section, we briefly discuss the methodology that we have used previously in \cite{Choong2020}. The all-orthogonality conditions of three qubits are given as
\begin{align}
\bar{t}_{111} t_{211} + \bar{t}_{121} t_{221} + \bar{t}_{112} t_{212} + \bar{t}_{122} t_{222} & = 0, \label{ThreeQubitsAllOrthogonality1} \\
\bar{t}_{111} t_{121} + \bar{t}_{211} t_{221} + \bar{t}_{112} t_{122} + \bar{t}_{212} t_{222} & = 0, \label{ThreeQubitsAllOrthogonality2} \\
\bar{t}_{111} t_{112} + \bar{t}_{211} t_{212} + \bar{t}_{121} t_{122} + \bar{t}_{221} t_{222} & = 0. \label{ThreeQubitsAllOrthogonality3}
\end{align}
By writing $\bar{t}_{111}$ and $t_{222}$ in terms of other variables,
\begin{align}
& t_{111} = - \frac{\bar{t}_{221} (t_{121} t_{212} - t_{122} t_{211}) + t_{112} (\left| t_{212} \right|^2 - \left| t_{122} \right|^2)}{t_{212} \bar{t}_{211} - t_{122} \bar{t}_{121}}, \label{t111} \\
& t_{222} = \frac{\bar{t}_{112} (t_{121} t_{212} - t_{122} t_{211}) + t_{221} (\left| t_{121} \right|^2 - \left| t_{211} \right|^2)}{\bar{t}_{212} t_{211} - \bar{t}_{122} t_{121}}, \label{t222}
\end{align}
we obtain
\begin{align}
& (\bar{t}_{221} t_{121} - \bar{t}_{212} t_{112}) (\bar{t}_{112} t_{212} + \bar{t}_{121} t_{221}) \nonumber \\ & \qquad + (\bar{t}_{122} t_{112} - \bar{t}_{221} t_{211}) (\bar{t}_{211} t_{221} + \bar{t}_{112} t_{122}) \nonumber \\ & \qquad + (\bar{t}_{212} t_{211} - \bar{t}_{122} t_{121}) (\bar{t}_{211} t_{212} + \bar{t}_{121} t_{122}) = 0. \label{ThreeQubitPencil}
\end{align}
After expanding equation (\ref{ThreeQubitPencil}), it is possible to separate the real and imaginary parts,
\begin{align}
& \left| t_{112} \right|^2 (\left| t_{122} \right|^2 - \left| t_{212} \right|^2) + \left| t_{121} \right|^2 (\left| t_{221} \right|^2 - \left| t_{122} \right|^2) \nonumber \\ & \qquad + \left| t_{211} \right|^2 (\left| t_{212} \right|^2 - \left| t_{221} \right|^2) = 0, \label{ThreeQubitsSub1} \\
& \bar{t}_{112} \bar{t}_{221} (t_{122} t_{211} - t_{121} t_{212}) + \bar{t}_{121} \bar{t}_{212} (t_{112} t_{221} - t_{122} t_{211}) \nonumber \\ & \qquad + \bar{t}_{122} \bar{t}_{211} (t_{121} t_{212} - t_{112} t_{221}) = 0. \label{ThreeQubitsSub2}
\end{align}
Equation (\ref{ThreeQubitsSub1}) is the basis to our previous work since it provides explicit relationship between the first $n$-mode singular values, $\sigma_1^{(i)2}$, i.e.
\begin{align}
& \left| t_{112} \right|^2 \left[ \sigma_1^{(1)2} - \sigma_1^{(2)2} \right] + \left| t_{211} \right|^2 \left[ \sigma_1^{(2)2} - \sigma_1^{(3)2} \right] + \left| t_{121} \right|^2 \left[ \sigma_1^{(3)2} - \sigma_1^{(1)2} \right] \nonumber \\ & \qquad = 0. \label{ThreeQubitsFinalAllOrtho}
\end{align}

On the other hand, equation (\ref{ThreeQubitsSub2}) fixes a relative phase of the three-qubit states. Since our results are based on the first $n$-mode singular values $\sigma_1^{(i)2}$, the relative phase does not affect our results. \newline
\textbf{Example:} Consider the following state $\left| \psi_1 \right\rangle$,
\begin{align*}
\left| \psi_1 \right\rangle & = t_{111} \left| 111 \right\rangle + t_{112} \left| 112 \right\rangle + t_{112} \left| 121 \right\rangle + t_{122} \left| 122 \right\rangle \\ & \quad + t_{211} \left| 211 \right\rangle + t_{212} \left| 212 \right\rangle + t_{212} \left| 221 \right\rangle + t_{222} \left| 222 \right\rangle,
\end{align*}
where it satisfies one of the bi-separable conditions $A|BC$, $t_{121} t_{212} = t_{112} t_{221}$. Hence, equation (\ref{ThreeQubitsSub2}) is satisfied. The all-orthogonality conditions are
\begin{align*}
\bar{t}_{111} t_{211} + 2 \bar{t}_{112} t_{212} + \bar{t}_{122} t_{222} & = 0, \\
\bar{t}_{111} t_{112} + \bar{t}_{211} t_{212} + \bar{t}_{112} t_{122} + \bar{t}_{212} t_{222} & = 0.
\end{align*}
The state $\left| \psi_1 \right\rangle$ has the same property ($\sigma_1^{(2)2} = \sigma_1^{(3)2} \neq \sigma_1^{(1)2}$) as the Slice state $\left| \text{S}_3 \right\rangle$,
\begin{align*}
\left| \text{S}_3 \right\rangle & = t_{111} \left| 111 \right\rangle + t_{122} \left| 122 \right\rangle + t_{211} \left| 211 \right\rangle + t_{222} \left| 222 \right\rangle
\end{align*}
with all-orthogonality condition
\begin{align*}
\bar{t}_{111} t_{211} + \bar{t}_{122} t_{222} = 0.
\end{align*}
Under a coarser classification procedure provided by equation (\ref{ThreeQubitsFinalAllOrtho}), they belong to the same family of states.

\subsection{Concurrency of three lines}

Now, let $L_1,\, L_2,\, L_3$ to be three lines intersecting at one point $(x,\, y)$,
\begin{align}
L_1 & \equiv a_1 x + b_1 y + c_1 = 0, \\
L_2 & \equiv a_2 x + b_2 y + c_2 = 0, \\
L_3 & \equiv a_3 x + b_3 y + c_3 = 0,
\end{align}
where $a_i,\, b_i,\, c_i$ for $i = 1,\, 2,\, 3$ are some coefficients and $x,\, y$ are indeterminates. In order to find the solution $(x,\, y)$ to the set of lines, we can substitute $x$ from $L_1$ and $y$ from $L_2$ into $L_3$ to get
\begin{align}
a_3 (b_1 c_2 - b_2 c_1) + b_3 (a_2 c_1 - a_1 c_2) + c_3 (a_1 b_2 - a_2 b_1) = 0. \label{ConcurrencyThreeLines}
\end{align}
Equation (\ref{ConcurrencyThreeLines}) can be written into a more concise form as
\begin{align}
\begin{vmatrix}
a_1 & b_1 & c_1 \\
a_2 & b_2 & c_2 \\
a_3 & b_3 & c_3
\end{vmatrix} = 0, \label{ConcurrencyThreeLinesDeterminant}
\end{align}
which is called the concurrency of three lines \cite{Pedoe1970}. By comparison, it is obvious that the derivation in Section \ref{SubSecClassificationThreeQubits} is the same as the concurrency of three lines, with $(x,\,y) = (\bar{t}_{111},\,t_{222})$. Since we did not specify the underlying field when deriving equation (\ref{ConcurrencyThreeLines}), the concurrency of three lines can be applied to complex field as long as the inherent properties of $x$ and $y$ (i.e. complex conjugate of $x$ and $y$) are not being used while solving the set of equations algebraically \cite{Todd1947}.

The biggest advantage in using the concurrency of three lines is that it is easier to find the solutions in the determinant form (\ref{ConcurrencyThreeLinesDeterminant}) in contrary to the polynomial form (\ref{ConcurrencyThreeLines}). There are two ways for a determinant to be zero,
\begin{enumerate}
\item At least one row (column) of the determinant is zero.

\item At least one row (column) of the determinant is linearly dependent to the other row (column).
\end{enumerate}

However, the linear dependence between rows (columns) of a determinant can always be decomposed into a combination of the former scenario, i.e. one row (column) of the determinant is zero. For instance, the linear dependence $L_1 = k'_2 L_2 + k'_3 L_3$, where $k'_i = - \frac{k_i}{k_1}$ and $i = 2,\, 3$ can be written as
\begin{align*}
\begin{vmatrix}
a_1 & b_1 & c_1 \\
a_2 & b_2 & c_2 \\
a_3 & b_3 & c_3
\end{vmatrix} & = \begin{vmatrix}
k'_2 a_2 & k'_2 b_2 & k'_2 c_2 \\
a_2 & b_2 & c_2 \\
a_3 & b_3 & c_3
\end{vmatrix} + \begin{vmatrix}
k'_3 a_3 & k'_3 b_3 & k'_3 c_3 \\
a_2 & b_2 & c_2 \\
a_3 & b_3 & c_3
\end{vmatrix} \\
& = \begin{vmatrix}
0 & 0 & 0 \\
a_2 & b_2 & c_2 \\
a_3 & b_3 & c_3
\end{vmatrix} + \begin{vmatrix}
0 & 0 & 0 \\
a_2 & b_2 & c_2 \\
a_3 & b_3 & c_3
\end{vmatrix} \\
& = \begin{vmatrix}
k'_2 a_2 & k'_2 b_2 & k'_2 c_2 \\
0 & 0 & 0 \\
a_3 & b_3 & c_3
\end{vmatrix} + \begin{vmatrix}
k'_3 a_3 & k'_3 b_3 & k'_3 c_3 \\
a_2 & b_2 & c_2 \\
0 & 0 & 0
\end{vmatrix} = \ldots
\end{align*}
Therefore, we will be able to identify all unique solutions to equation (\ref{ConcurrencyThreeLinesDeterminant}) by studying only the former scenario.

Furthermore, in order to simplify the computational process, we focus only on the minimum requirement for a determinant to be zero, i.e. when one row (column) of the determinant is zero. This consideration does not generate generalized GHZ states, however it can be recognized as
\begin{align}
\left| \text{GHZ} \right\rangle = t_{1 \ldots 1} \left| 1 \ldots 1 \right\rangle + t_{2 \ldots 2} \left| 2 \ldots 2 \right\rangle.
\end{align}

\subsection{Conjugate concurrent variables}

Since the inherent properties of $x$ and $y$ cannot be used while solving the all-orthogonality conditions, we can make use of the concurrency of three lines. We formalize this idea with the following definition.

\begin{definition}[Conjugate concurrent variables]
Let $\{L_i\}$ be a set of all-orthogonality conditions. A pair of conjugate concurrent variables $(x,\,y)$ satisfies the following two criteria:-
\begin{enumerate}
\item The relative phase between the conjugate concurrent variables is preserved throughout the all-orthogonality conditions;

\item The pair of conjugate concurrent variables must exist in every all-orthogonality conditions.
\end{enumerate}
\end{definition}

The first criterion is stated so that we do not make use of the inherent properties of the conjugate concurrent variables (CCV). From equations (\ref{ThreeQubitsAllOrthogonality1}) to (\ref{ThreeQubitsAllOrthogonality3}), there are four pairs of CCV: $(\bar{t}_{111},\,t_{222})$, $(\bar{t}_{112},\,t_{221})$, $(\bar{t}_{121},\,t_{212})$ and $(\bar{t}_{122},\,t_{211})$. Meanwhile, $(\bar{t}_{111},\,\bar{t}_{112})$ is not a pair of CCV because the relative phase between $\bar{t}_{111}$ and $\bar{t}_{112}$ changes in equation (\ref{ThreeQubitsAllOrthogonality3}). One has to make use of the inherent property of $\bar{t}_{112}$ as a complex variable to be able to solve the all-orthogonality conditions.

The second criterion is required so that the solutions that we found will satisfy every all-orthogonality conditions. This implies that that the current approach is limited to multi-qubit systems. For instance, if we consider the all-orthogonality conditions of a $(2 \times 2 \times 3)$-system,
\begin{align*}
\bar{t}_{111} t_{211} + \bar{t}_{112} t_{212} + \bar{t}_{113} t_{213} + \bar{t}_{121} t_{221} + \bar{t}_{122} t_{222} + \bar{t}_{123} t_{223} & = 0, \\
\bar{t}_{111} t_{121} + \bar{t}_{211} t_{221} + \bar{t}_{112} t_{122} + \bar{t}_{212} t_{222} + \bar{t}_{113} t_{123} + \bar{t}_{213} t_{223} & = 0, \\
\bar{t}_{111} t_{112} + \bar{t}_{121} t_{122} + \bar{t}_{211} t_{212} + \bar{t}_{221} t_{222} & = 0, \\
\bar{t}_{111} t_{113} + \bar{t}_{121} t_{123} + \bar{t}_{211} t_{213} + \bar{t}_{221} t_{223} & = 0, \\
\bar{t}_{112} t_{113} + \bar{t}_{122} t_{123} + \bar{t}_{212} t_{213} + \bar{t}_{222} t_{223} & = 0,
\end{align*}
we can see that some of the variables do not exist in every all-orthogonality conditions. Therefore, we say that a pair of CCV does not exist in this $(2 \times 2 \times 3)$-system.

\subsection{Special three-qubit core tensors by concurrency of three lines}

From equations (\ref{ThreeQubitsAllOrthogonality1}) to (\ref{ThreeQubitsAllOrthogonality3}), the concurrency of three lines for all-orthogonality conditions of three qubits is given by
\begin{align}
\begin{vmatrix}
t_{211} & \bar{t}_{122} & \bar{t}_{121} t_{221} + \bar{t}_{112} t_{212} \\
t_{121} & \bar{t}_{212} & \bar{t}_{211} t_{221} + \bar{t}_{112} t_{122} \\
t_{112} & \bar{t}_{221} & \bar{t}_{211} t_{212} + \bar{t}_{121} t_{122}
\end{vmatrix} = 0, \label{ConcurrencyAllOrthogonality}
\end{align}
where $(x,\, y) = (\bar{t}_{111}, t_{222})$. The details of our calculations will be shown in Appendix \ref{AppendixConcurrencyThreeQubits}. The results are summarized in Table \ref{TableConcurrencyThreeQubits}.

From Table \ref{TableConcurrencyThreeQubits}, we can see that by considering only the minimum requirements to satisfy equation (\ref{ConcurrencyAllOrthogonality}), it is enough to recover all the generic special states of three qubits besides the generalized GHZ states,
\begin{align}
\left| \text{GHZ} \right\rangle = t_{111} \left| 111 \right\rangle + t_{222} \left| 222 \right\rangle.
\end{align}
Some of the special states that we have identified are not generic because of the ordering property of higher order singular value decomposition (HOSVD). As an example, if we study the following three-qubit state,
\begin{align*}
& \left| \psi \right\rangle = t_{121} \left| 121 \right\rangle + t_{122} \left| 122 \right\rangle + t_{211} \left| 211 \right\rangle + t_{212} \left| 212 \right\rangle, \\
& \bar{t}_{211} t_{212} + \bar{t}_{121} t_{122} = 0,
\end{align*}
the first $n$-mode singular values are given as
\begin{align*}
\sigma_1^{(1)2} & = \left| t_{121} \right|^2 + \left| t_{122} \right|^2 = \sigma_2^{(2)2}, \\
\sigma_1^{(2)2} & = \left| t_{211} \right|^2 + \left| t_{212} \right|^2 = \sigma_2^{(1)2}, \\
\sigma_1^{(3)2} & = \left| t_{121} \right|^2 + \left| t_{211} \right|^2.
\end{align*}
Due to the ordering property, equality is possible only when $(\sigma_1^{(1)2},\,\sigma_1^{(2)2},\,\sigma_1^{(3)2}) = (\frac{1}{2},\, \frac{1}{2},\, \sigma_1^{(3)2})$. Therefore, it is not a generic special state of three qubits.

\scriptsize
\begin{longtable}{|p{0.4\linewidth}|p{0.6\linewidth}|}
\caption{Special three-qubit core tensors due to the concurrency of three lines for all-orthogonality conditions of three qubits} \label{TableConcurrencyThreeQubits} \\
\hline
Row (Column) checking & States \\
\hline
1. Column 1 = 0 & $\left| \text{B}_1 \right\rangle = t_{111} \left| 111 \right\rangle + t_{122} \left| 122 \right\rangle + t_{212} \left| 212 \right\rangle + t_{221} \left| 221 \right\rangle$ \\
\hline
2. Column 2 = 0 & $\left| \text{B}_2 \right\rangle = t_{112} \left| 112 \right\rangle + t_{121} \left| 121 \right\rangle + t_{211} \left| 211 \right\rangle + t_{222} \left| 222 \right\rangle$ \\
\hline
3. Column 3 = 0 (Non-generic) & (a) $t_{112} = t_{221} = 0$ \\ & \phantom{(a)} $\left| \psi \right\rangle = t_{121} \left| 121 \right\rangle + t_{122} \left| 122 \right\rangle + t_{211} \left| 211 \right\rangle + t_{212} \left| 212 \right\rangle,$ \\ & \phantom{(a)} $\bar{t}_{211} t_{212} + \bar{t}_{121} t_{122} = 0$ \\
\cline{2-2}
& (b) $t_{121} = t_{212} = 0$ \\ & \phantom{(b)} $\left| \psi \right\rangle = t_{112} \left| 112 \right\rangle + t_{122} \left| 122 \right\rangle + t_{211} \left| 211 \right\rangle + t_{221} \left| 221 \right\rangle,$ \\ & \phantom{(b)} $\bar{t}_{211} t_{221} + \bar{t}_{112} t_{122} = 0$ \\
\cline{2-2}
& (c) $t_{122} = t_{211} = 0$ \\ & \phantom{(c)} $\left| \psi \right\rangle = t_{112} \left| 112 \right\rangle + t_{121} \left| 121 \right\rangle + t_{212} \left| 212 \right\rangle + t_{221} \left| 221 \right\rangle,$ \\ & \phantom{(c)} $\bar{t}_{121} t_{221} + \bar{t}_{112} t_{212} = 0$ \\
\hline
4. Row 1 = 0 & (a) $t_{111} = t_{222} = 0$ \\ & \phantom{(a)} Same as 3(c) \\
\cline{2-2}
& (b) $t_{112} = t_{221} = 0$ \\ & \phantom{(b)} $\left| \text{S}_2 \right\rangle = t_{111} \left| 111 \right\rangle + t_{121} \left| 121 \right\rangle + t_{212} \left| 212 \right\rangle + t_{222} \left| 222 \right\rangle,$ \\ & \phantom{(b)} $\bar{t}_{111} t_{121} + \bar{t}_{212} t_{222} = 0$ \\
\cline{2-2}
& (c) $t_{121} = t_{212} = 0$ \\ & \phantom{(c)} $\left| \text{S}_1 \right\rangle = t_{111} \left| 111 \right\rangle + t_{112} \left| 112 \right\rangle + t_{221} \left| 221 \right\rangle + t_{222} \left| 222 \right\rangle,$ \\ & \phantom{(c)} $\bar{t}_{111} t_{112} + \bar{t}_{221} t_{222} = 0$ \\
\hline
5. Row 2 = 0 & (a) $t_{111} = t_{222} = 0$ \\ & \phantom{(a)} Same as 3(b) \\
\cline{2-2}
& (b) $t_{112} = t_{221} = 0$ \\ & \phantom{(b)} $\left| \text{S}_3 \right\rangle = t_{111} \left| 111 \right\rangle + t_{122} \left| 122 \right\rangle + t_{211} \left| 211 \right\rangle + t_{222} \left| 222 \right\rangle,$ \\ & \phantom{(b)} $\bar{t}_{111} t_{211} + \bar{t}_{122} t_{222} = 0$ \\
\cline{2-2}
& (c) $t_{121} = t_{212} = 0$ \\ & \phantom{(c)} Same as 4 (c) \\
\hline
6. Row 3 = 0 & (a) $t_{111} = t_{222} = 0$ \\ & \phantom{(a)} Same as 3(a) \\
\cline{2-2}
& (b) $t_{122} = t_{211} = 0$ \\ & \phantom{(b)} Same as 4 (b) \\
\cline{2-2}
& (c) $t_{121} = t_{212} = 0$ \\ & \phantom{(c)} Same as 5 (c) \\
\hline
\end{longtable}
\normalsize

\section{Four qubits and beyond} \label{SecFourQubitsAndGeneralization}

\subsection{Generalization to multi-qubit states} \label{SubSecGeneralization}

For multi-qubit states, we can generalize our approach by the following algorithm.
\begin{enumerate}
\item Select a pair of conjugate concurrent variables (CCV) and formulate the concurrency of three lines accordingly;

\item Perform row (column) checking on the concurrency of three lines;

\item For a system of all-orthogonality conditions without a pair of CCV, find its family of states;

\item For a system of all-orthogonality conditions with a pair of CCV, select another pair of CCV and formulate the next iteration of concurrency of three lines accordingly;

\item The process stops when at most two all-orthogonality conditions are left.
\end{enumerate}

There are $n$ number of all-orthogonality conditions for $n$-qubit states. In order to formulate the concurrency of three lines for the set of all-orthogonality conditions, we need to exhaust all the possible combinations between the $n$ number of all-orthogonality conditions. This is a combinatorial problem of selecting three out of $n$-th all-orthogonality conditions, therefore the number of simultaneous concurrency of three lines that we can form is given by $\frac{n!}{3!(n-3)!}$.

In order to explore all the minimum requirements for the set of concurrency of three lines to be true, we need to have at least $n-2$ number of rows to be zero during the row checking. This is another combinatorial problem of selecting $n-2$ out of $n$ rows, which requires $\frac{n!}{2!(n-2)!}$ of row checking in total for one iteration. For column checking, we always need three checks regardless of the number of simultaneous concurrency of three lines that we have.

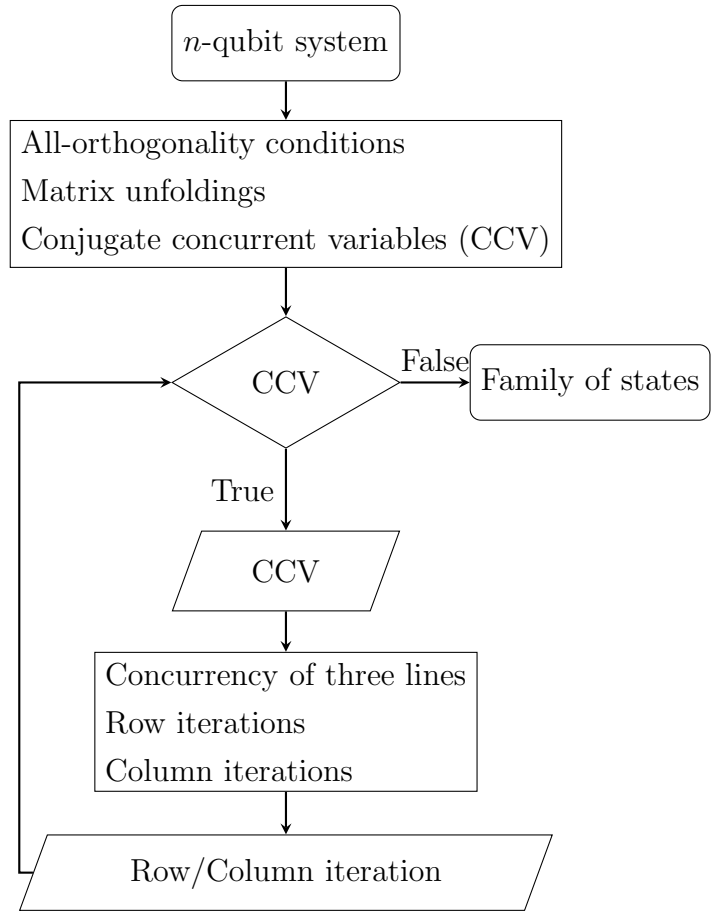
\begin{figure}[H]
\begin{center}
\begin{tikzpicture}[node distance=2cm]
\node(start)[startstop]{$n$-qubit system};
\node(pro1)[process,below of=start]{All-orthogonality conditions\\Matrix unfoldings\\Conjugate concurrent variables (CCV)};
\node(dec)[decision,below of=pro1,yshift=-0.5cm]{CCV};
\node(end)[startstop,right of=dec,xshift=2cm]{Family of states};
\node(in1)[io,below of=dec,yshift=-0.5cm]{CCV};
\node(pro2)[process,below of=in1]{Concurrency of three lines\\Row iterations\\Column iterations};
\node(in2)[io,below of=pro2]{Row/Column iteration};
\draw[arrow](start)--(pro1);
\draw[arrow](pro1)--(dec);
\draw[arrow](dec)--node[anchor=south]{False}(end);
\draw[arrow](dec)--node[anchor=east]{True}(in1);
\draw[arrow](in1)--(pro2);
\draw[arrow](pro2)--(in2);
\draw[arrow](in2.west)-|([xshift=-2cm]dec.west)--(dec.west);
\end{tikzpicture}
\caption{Flowchart for the identification of special multi-qubit states}
\end{center}
\end{figure}

Since we explore every requirements to satisfy the concurrency of three lines for the set of all-orthogonality conditions by going through several iterations, the choice of CCV does not matter.

\subsection{Special four-qubit core tensors by concurrency of three lines}

The all-orthogonality conditions for four qubits are given as
\begin{align}
& \bar{t}_{1111} t_{2111} + \bar{t}_{1112} t_{2112} + \bar{t}_{1121} t_{2121} + \bar{t}_{1122} t_{2122} + \bar{t}_{1211} t_{2211} + \bar{t}_{1212} t_{2212} \nonumber \\ & \qquad + \bar{t}_{1221} t_{2221} + \bar{t}_{1222} t_{2222} = 0, \label{FourQubitsAllOrthogonality1} \\
& \bar{t}_{1111} t_{1211} + \bar{t}_{2111} t_{2211} + \bar{t}_{1112} t_{1212} + \bar{t}_{2112} t_{2212} + \bar{t}_{1121} t_{1221} + \bar{t}_{2121} t_{2221} \nonumber \\ & \qquad + \bar{t}_{1122} t_{1222} + \bar{t}_{2122} t_{2222} = 0, \label{FourQubitsAllOrthogonality2} \\
& \bar{t}_{1111} t_{1121} + \bar{t}_{1211} t_{1221} + \bar{t}_{2111} t_{2121} + \bar{t}_{2211} t_{2221} + \bar{t}_{1112} t_{1122} + \bar{t}_{1212} t_{1222} \nonumber \\ & \qquad + \bar{t}_{2112} t_{2122} + \bar{t}_{2212} t_{2222} = 0, \label{FourQubitsAllOrthogonality3} \\
& \bar{t}_{1111} t_{1112} + \bar{t}_{1121} t_{1122} + \bar{t}_{1211} t_{1212} + \bar{t}_{1221} t_{1222} + \bar{t}_{2111} t_{2112} + \bar{t}_{2121} t_{2122} \nonumber \\ & \qquad + \bar{t}_{2211} t_{2212} + \bar{t}_{2221} t_{2222} = 0. \label{FourQubitsAllOrthogonality4}
\end{align}

We can formulate four concurrency of three lines from equations (\ref{FourQubitsAllOrthogonality1}) to (\ref{FourQubitsAllOrthogonality4}). By selecting $\bar{t}_{1111}$ and $t_{2222}$ as the pair of conjugate concurrent variables (CCV), the first iteration is given by
\begin{align}
\begin{vmatrix}
t_{2111} & \bar{t}_{1222} & c_1 \\
t_{1211} & \bar{t}_{2122} & c_2 \\
t_{1121} & \bar{t}_{2212} & c_3
\end{vmatrix} = 0, \label{ConcurrencyFourQubits1} \\
\begin{vmatrix}
t_{2111} & \bar{t}_{1222} & c_1 \\
t_{1211} & \bar{t}_{2122} & c_2 \\
t_{1112} & \bar{t}_{2221} & c_4
\end{vmatrix} = 0, \\
\begin{vmatrix}
t_{2111} & \bar{t}_{1222} & c_1 \\
t_{1121} & \bar{t}_{2212} & c_3 \\
t_{1112} & \bar{t}_{2221} & c_4
\end{vmatrix} = 0, \\
\begin{vmatrix}
t_{1211} & \bar{t}_{2122} & c_2 \\
t_{1121} & \bar{t}_{2212} & c_3 \\
t_{1112} & \bar{t}_{2221} & c_4
\end{vmatrix} = 0, \label{ConcurrencyFourQubits4}
\end{align}
where
\begin{align}
c_1 & = \bar{t}_{1112} t_{2112} + \bar{t}_{1121} t_{2121} + \bar{t}_{1122} t_{2122} + \bar{t}_{1211} t_{2211} + \bar{t}_{1212} t_{2212} + \bar{t}_{1221} t_{2221}, \\
c_2 & = \bar{t}_{2111} t_{2211} + \bar{t}_{1112} t_{1212} + \bar{t}_{2112} t_{2212} + \bar{t}_{1121} t_{1221} + \bar{t}_{2121} t_{2221} + \bar{t}_{1122} t_{1222}, \\
c_3 & = \bar{t}_{1211} t_{1221} + \bar{t}_{2111} t_{2121} + \bar{t}_{2211} t_{2221} + \bar{t}_{1112} t_{1122} + \bar{t}_{1212} t_{1222} + \bar{t}_{2112} t_{2122}, \\
c_4 & = \bar{t}_{1121} t_{1122} + \bar{t}_{1211} t_{1212} + \bar{t}_{1221} t_{1222} + \bar{t}_{2111} t_{2112} + \bar{t}_{2121} t_{2122} + \bar{t}_{2211} t_{2212}.
\end{align}

As mentioned in Section \ref{SubSecGeneralization}, we need to allow two rows to be zero in order to minimally satisfy the simultaneous concurrency of three lines from equations (\ref{ConcurrencyFourQubits1}) to (\ref{ConcurrencyFourQubits4}). We need to perform six row-checkings, i.e. rows $(1\text{-}2)$, $(1\text{-}3)$, $(1\text{-}4)$, $(2\text{-}3)$, $(2\text{-}4)$ and $(3\text{-}4)$. The number of column checking we need to perform is 3. To summarize, we performed a total of 73 of row and column checking across 3 iterations for four qubits. We used Mathematica to perform all the computational tasks.

Our results can be summarized as follows:-

\scriptsize
\begin{longtable}{|p{0.32\linewidth}|p{0.68\linewidth}|}
\caption{Special four-qubit core tensors due to the concurrency of three lines for all-orthogonality conditions} \\
\hline
Cases & States \\
\hline
1. $\sigma_1^{(1)2} \neq \sigma_1^{(2)2} \neq \sigma_1^{(3)2} \neq \sigma_1^{(4)2}$ & $\left| \psi \right\rangle = t_{1111} \left| 1111 \right\rangle + t_{1122} \left| 1122 \right\rangle + t_{1212} \left| 1212 \right\rangle + t_{1221} \left| 1221 \right\rangle + t_{2112} \left| 2112 \right\rangle$ \\ & \phantom{$\left| \psi \right\rangle =$ } $+ t_{2121} \left| 2121 \right\rangle + t_{2211} \left| 2211 \right\rangle + t_{2222} \left| 2222 \right\rangle$ \\
& $\left| \psi \right\rangle = t_{1112} \left| 1112 \right\rangle + t_{1121} \left| 1121 \right\rangle + t_{1211} \left| 1211 \right\rangle + t_{1222} \left| 1222 \right\rangle + t_{2111} \left| 2111 \right\rangle$ \\ & \phantom{$\left| \psi \right\rangle =$ } $+ t_{2122} \left| 2122 \right\rangle + t_{2212} \left| 2212 \right\rangle + t_{2221} \left| 2221 \right\rangle$ \\
& $\left| \psi \right\rangle = t_{1122} \left| 1122 \right\rangle + t_{1212} \left| 1212 \right\rangle + t_{1221} \left| 1221 \right\rangle + t_{2112} \left| 2112 \right\rangle + t_{2121} \left| 2121 \right\rangle$ \\ & \phantom{$\left| \psi \right\rangle =$ } $+ t_{2211} \left| 2211 \right\rangle$ \\
& $\left| \psi \right\rangle = t_{1111} \left| 1111 \right\rangle + t_{1122} \left| 1122 \right\rangle + t_{1212} \left| 1212 \right\rangle + t_{2112} \left| 2112 \right\rangle + t_{2221} \left| 2221 \right\rangle$ \\
& $\left| \psi \right\rangle = t_{1111} \left| 1111 \right\rangle + t_{1122} \left| 1122 \right\rangle + t_{1221} \left| 1221 \right\rangle + t_{2121} \left| 2121 \right\rangle + t_{2212} \left| 2212 \right\rangle$ \\
& $\left| \psi \right\rangle = t_{1111} \left| 1111 \right\rangle + t_{1212} \left| 1212 \right\rangle + t_{1221} \left| 1221 \right\rangle + t_{2122} \left| 2122 \right\rangle + t_{2211} \left| 2211 \right\rangle$ \\
& $\left| \psi \right\rangle = t_{1112} \left| 1112 \right\rangle + t_{1211} \left| 1211 \right\rangle + t_{1222} \left| 1222 \right\rangle + t_{2121} \left| 2121 \right\rangle + t_{2212} \left| 2212 \right\rangle$ \\
& $\left| \psi \right\rangle = t_{1121} \left| 1121 \right\rangle + t_{1212} \left| 1212 \right\rangle + t_{2111} \left| 2111 \right\rangle + t_{2122} \left| 2122 \right\rangle + t_{2221} \left| 2221 \right\rangle$ \\
& $\left| \psi \right\rangle = t_{1112} \left| 1112 \right\rangle + t_{1221} \left| 1221 \right\rangle + t_{2121} \left| 2121 \right\rangle + t_{2211} \left| 2211 \right\rangle + t_{2222} \left| 2222 \right\rangle$ \\
& $\left| \psi \right\rangle = t_{1121} \left| 1121 \right\rangle + t_{1212} \left| 1212 \right\rangle + t_{2112} \left| 2112 \right\rangle + t_{2211} \left| 2211 \right\rangle + t_{2222} \left| 2222 \right\rangle$ \\
& $\left| \psi \right\rangle = t_{1122} \left| 1122 \right\rangle + t_{1211} \left| 1211 \right\rangle + t_{2112} \left| 2112 \right\rangle + t_{2121} \left| 2121 \right\rangle + t_{2222} \left| 2222 \right\rangle$ \\
\hline
\pagebreak
\hline
2. $\sigma_1^{(i)2} = \sigma_1^{(j)2},\, i \neq j$ & (a) $\sigma_1^{(1)2} = \sigma_1^{(2)2}$ \\ & \phantom{(a)} $\left| \psi \right\rangle = t_{1111} \left| 1111 \right\rangle + t_{1112} \left| 1112 \right\rangle + t_{1121} \left| 1121 \right\rangle + t_{1122} \left| 1122 \right\rangle$ \\ & \phantom{(a) $\left| \psi \right\rangle =$ } $+ t_{2211} \left| 2211 \right\rangle + t_{2212} \left| 2212 \right\rangle + t_{2221} \left| 2221 \right\rangle + t_{2222} \left| 2222 \right\rangle,$ \\ & \phantom{(a)} $\bar{t}_{1111} t_{1121} + \bar{t}_{1112} t_{1122} + \bar{t}_{2211} t_{2221} + \bar{t}_{2212} t_{2222} = 0,$ \\ & \phantom{(a)} $\bar{t}_{1111} t_{1112} + \bar{t}_{1121} t_{1122} + \bar{t}_{2211} t_{2212} + \bar{t}_{2221} t_{2222} = 0$ \\
\cline{2-2}
& (b) $\sigma_1^{(1)2} = \sigma_1^{(3)2}$ \\ & \phantom{(b)} $\left| \psi \right\rangle = t_{1111} \left| 1111 \right\rangle + t_{1112} \left| 1112 \right\rangle + t_{1211} \left| 1211 \right\rangle + t_{1212} \left| 1212 \right\rangle$ \\ & \phantom{(b) $\left| \psi \right\rangle =$ } $+ t_{2121} \left| 2121 \right\rangle + t_{2122} \left| 2122 \right\rangle + t_{2221} \left| 2221 \right\rangle + t_{2222} \left| 2222 \right\rangle,$ \\ & \phantom{(b)} $\bar{t}_{1111} t_{1211} + \bar{t}_{1112} t_{1212} + \bar{t}_{2121} t_{2221} + \bar{t}_{2122} t_{2222} = 0,$ \\ & \phantom{(b)} $\bar{t}_{1111} t_{1112} + \bar{t}_{1211} t_{1212} + \bar{t}_{2121} t_{2122} + \bar{t}_{2221} t_{2222} = 0$ \\
\cline{2-2}
& (c) $\sigma_1^{(1)2} = \sigma_1^{(4)2}$ \\ & \phantom{(c)} $\left| \psi \right\rangle = t_{1111} \left| 1111 \right\rangle + t_{1121} \left| 1121 \right\rangle + t_{1211} \left| 1211 \right\rangle + t_{1221} \left| 1221 \right\rangle$ \\ & \phantom{(c) $\left| \psi \right\rangle =$ } $+ t_{2112} \left| 2112 \right\rangle + t_{2122} \left| 2122 \right\rangle + t_{2212} \left| 2212 \right\rangle + t_{2222} \left| 2222 \right\rangle,$ \\ & \phantom{(c)} $\bar{t}_{1111} t_{1211} + \bar{t}_{1121} t_{1221} + \bar{t}_{2112} t_{2212} + \bar{t}_{2122} t_{2222} = 0,$ \\ & \phantom{(c)} $\bar{t}_{1111} t_{1121} + \bar{t}_{1211} t_{1221} + \bar{t}_{2112} t_{2122} + \bar{t}_{2212} t_{2222} = 0$ \\
\cline{2-2}
& (d) $\sigma_1^{(2)2} = \sigma_1^{(3)2}$ \\ & \phantom{(d)} $\left| \psi \right\rangle = t_{1111} \left| 1111 \right\rangle + t_{1112} \left| 1112 \right\rangle + t_{1221} \left| 1221 \right\rangle + t_{1222} \left| 1222 \right\rangle$ \\ & \phantom{(d) $\left| \psi \right\rangle =$ } $+ t_{2111} \left| 2111 \right\rangle + t_{2112} \left| 2112 \right\rangle + t_{2221} \left| 2221 \right\rangle + t_{2222} \left| 2222 \right\rangle,$ \\ & \phantom{(d)} $\bar{t}_{1111} t_{2111} + \bar{t}_{1112} t_{2112} + \bar{t}_{1221} t_{2221} + \bar{t}_{1222} t_{2222} = 0,$ \\ & \phantom{(d)} $\bar{t}_{1111} t_{1112} + \bar{t}_{1221} t_{1222} + \bar{t}_{2111} t_{2112} + \bar{t}_{2221} t_{2222} = 0$ \\
\cline{2-2}
& (e) $\sigma_1^{(2)2} = \sigma_1^{(4)2}$ \\ & \phantom{(e)} $\left| \psi \right\rangle = t_{1111} \left| 1111 \right\rangle + t_{1121} \left| 1121 \right\rangle + t_{1212} \left| 1212 \right\rangle + t_{1222} \left| 1222 \right\rangle$ \\ & \phantom{(e) $\left| \psi \right\rangle =$ } $+ t_{2111} \left| 2111 \right\rangle + t_{2121} \left| 2121 \right\rangle + t_{2212} \left| 2212 \right\rangle + t_{2222} \left| 2222 \right\rangle,$ \\ & \phantom{(e)} $\bar{t}_{1111} t_{2111} + \bar{t}_{1121} t_{2121} + \bar{t}_{1212} t_{2212} + \bar{t}_{1222} t_{2222} = 0,$ \\ & \phantom{(e)} $\bar{t}_{1111} t_{1121} + \bar{t}_{1212} t_{1222} + \bar{t}_{2111} t_{2121} + \bar{t}_{2212} t_{2222} = 0$ \\
\cline{2-2}
& (f) $\sigma_1^{(3)2} = \sigma_1^{(4)2}$ \\ & \phantom{(f)} $\left| \psi \right\rangle = t_{1111} \left| 1111 \right\rangle + t_{1122} \left| 1122 \right\rangle + t_{1211} \left| 1211 \right\rangle + t_{1222} \left| 1222 \right\rangle$ \\ & \phantom{(f) $\left| \psi \right\rangle =$ } $+ t_{2111} \left| 2111 \right\rangle + t_{2122} \left| 2122 \right\rangle + t_{2211} \left| 2211 \right\rangle + t_{2222} \left| 2222 \right\rangle,$ \\ & \phantom{(f)} $\bar{t}_{1111} t_{2111} + \bar{t}_{1122} t_{2122} + \bar{t}_{1211} t_{2211} + \bar{t}_{1222} t_{2222} = 0,$ \\ & \phantom{(f)} $\bar{t}_{1111} t_{1211} + \bar{t}_{1122} t_{1222} + \bar{t}_{2111} t_{2211} + \bar{t}_{2122} t_{2222} = 0$ \\
\hline
3. $\sigma_1^{(i)2} = \sigma_1^{(j)2},\, \sigma_1^{(k)2} = \sigma_1^{(l)2},$ & (a) $\sigma_1^{(1)2} = \sigma_1^{(2)2},\, \sigma_1^{(3)2} = \sigma_1^{(4)2}$ \\ \phantom{3.} $i \neq j \neq k \neq l$ & \phantom{(a)} $\left| \psi \right\rangle = t_{1111} \left| 1111 \right\rangle + t_{1122} \left| 1122 \right\rangle + t_{2211} \left| 2211 \right\rangle + t_{2222} \left| 2222 \right\rangle$ \\
\cline{2-2}
& (b) $\sigma_1^{(1)2} = \sigma_1^{(3)2},\, \sigma_1^{(2)2} = \sigma_1^{(4)2}$ \\ & \phantom{(b)} $\left| \psi \right\rangle = t_{1111} \left| 1111 \right\rangle + t_{1212} \left| 1212 \right\rangle + t_{2121} \left| 2121 \right\rangle + t_{2222} \left| 2222 \right\rangle$ \\
\cline{2-2}
& (c) $\sigma_1^{(1)2} = \sigma_1^{(4)2},\, \sigma_1^{(2)2} = \sigma_1^{(3)2}$ \\ & \phantom{(c)} $\left| \psi \right\rangle = t_{1111} \left| 1111 \right\rangle + t_{1221} \left| 1221 \right\rangle + t_{2112} \left| 2112 \right\rangle + t_{2222} \left| 2222 \right\rangle$ \\
\hline
4. $\sigma_1^{(i)2} = \sigma_1^{(j)2} = \sigma_1^{(k)2},$ & (a) $\sigma_1^{(1)2} = \sigma_1^{(2)2} = \sigma_1^{(3)2}$ \\ \phantom{4.} $i \neq j \neq k$ & \phantom{(a)} $\left| \psi \right\rangle = t_{1111} \left| 1111 \right\rangle + t_{1112} \left| 1112 \right\rangle + t_{2221} \left| 2221 \right\rangle + t_{2222} \left| 2222 \right\rangle,$ \\ & \phantom{(a)} $\bar{t}_{1112} t_{1111} + \bar{t}_{2222} t_{2221} = 0$ \\
\cline{2-2}
& (b) $\sigma_1^{(1)2} = \sigma_1^{(2)2} = \sigma_1^{(4)2}$ \\ & \phantom{(b)} $\left| \psi \right\rangle = t_{1111} \left| 1111 \right\rangle + t_{1121} \left| 1121 \right\rangle + t_{2212} \left| 2212 \right\rangle + t_{2222} \left| 2222 \right\rangle,$ \\ & \phantom{(b)} $\bar{t}_{1111} t_{1121} + \bar{t}_{2212} t_{2222} = 0$ \\
\cline{2-2}
& (c) $\sigma_1^{(1)2} = \sigma_1^{(3)2} = \sigma_1^{(4)2}$ \\ & \phantom{(c)} $\left| \psi \right\rangle = t_{1111} \left| 1111 \right\rangle + t_{1211} \left| 1211 \right\rangle + t_{2122} \left| 2122 \right\rangle + t_{2222} \left| 2222 \right\rangle,$ \\ & \phantom{(c)} $\bar{t}_{1111} t_{1211} + \bar{t}_{2122} t_{2222} = 0$ \\
\cline{2-2}
& (d) $\sigma_1^{(2)2} = \sigma_1^{(3)2} = \sigma_1^{(4)2}$ \\ & \phantom{(d)} $\left| \psi \right\rangle = t_{1111} \left| 1111 \right\rangle + t_{1222} \left| 1222 \right\rangle + t_{2111} \left| 2111 \right\rangle + t_{2222} \left| 2222 \right\rangle,$ \\ & \phantom{(d)} $\bar{t}_{1111} t_{2111} + \bar{t}_{1222} t_{2222} = 0$ \\
\hline
\end{longtable}
\normalsize

As expected, we did not recover the generalized GHZ states of four qubits by using this approach, however it is recognized as
\begin{align}
\left| \text{GHZ} \right\rangle = t_{1111} \left| 1111 \right\rangle + t_{2222} \left| 2222 \right\rangle,
\end{align}
with the first $n$-mode singular values $\sigma_1^{(1)2} = \sigma_1^{(2)2} = \sigma_1^{(3)2} = \sigma_1^{(4)2}$.

\section{Conclusion}

In this work, we discussed how the matrix unfolding of a multipartite state is related to its reduced density matrices, and how the higher order singular value decomposition (HOSVD) is related to the simultaneous diagonalization of one-body reduced density matrices (RDM). While these results are not new and have been discussed in the past (for example, the $A \text{-} BC$, $B \text{-} AC$ and $C \text{-} AB$ bipartite decomposition of three qubits \cite{Albeverio2005}; tensor flattening \cite{Gharani2021}; trace decomposition \cite{Kraus2010A,Kraus2010B}; simultaneous diagonalization of one-body RDM due to the momentum map and Cartan subalgebra \cite{Sawicki2011A}), we showed these results from the perspectives of matrix unfolding and HOSVD.

From our previous work \cite{Choong2020}, we solved for the solutions to the set of all-orthogonality conditions for three qubits and obtained some results equivalent to the local unitary (LU) classification of three qubits \cite{Carteret2000}. Stemming from the same methodology, we proposed a simpler coarse-grained method to identify special multi-qubit core tensors by using the concurrency of three lines. A detailed study on the special core tensors based on their entanglement and geometrical properties is an interesting future direction we wish to pursue.

\section{Acknowledgement}

This research was supported by Fundamental Research Grant Scheme (FRGS) funded by Ministry of Higher Education of Malaysia with reference code FRGS/1/2019/STG02/UPM/02/3. The first author is sponsored by Ministry of Education (MOE) under MyBrainSc.

\bibliographystyle{ieeetr}
\bibliography{PaperRef}

\newpage
\appendix
\section{Proofs to Proposition \ref{MatrixUnfoldingRDMProposition} and Theorem \ref{HOSVDReducedDensityMatrixTheorem}}

\subsection{Proposition \ref{MatrixUnfoldingRDMProposition}} \label{AppendixProofMatrixUnfolding}
\setcounter{proposition}{0}
\begin{proposition}[Matrix unfolding and reduced density matrices]
The $n$-th matrix unfolding $\Psi_{(n)}$ of an $N$-th order tensor $\Psi \in \mathbb{C}^{I_1} \otimes \ldots \otimes \mathbb{C}^{I_n} \otimes \ldots \otimes \mathbb{C}^{I_N}$ is related to its one-body and $(n-1)$-body reduced density matrices through the following relations respectively,
\begin{align}
\Psi_{(n)} \Psi_{(n)}^\dagger & = \rho_n, \\
\Psi_{(n)}^\text{T} \bar{\Psi}_{(n)} & = \rho_{n+1\, \ldots N\, 1\, \ldots n-1}.
\end{align}
\end{proposition}

\begin{proof}
First, we generalize the partial trace operation to multipartite states
\begin{align}
& \text{Tr}_{n} (\left| i_1 \ldots i_N \right\rangle \left\langle j_1 \ldots j_N \right|) \nonumber \\ & \quad = \left| i_1 \ldots i_{n-1} i_{n+1} \ldots i_N \right\rangle \left\langle j_1 \ldots j_{n-1} j_{n+1} \ldots j_N \right| \text{Tr} (\left| i_n \right\rangle \left\langle j_n \right|) \nonumber \\ & \quad = \left| i_1 \ldots i_{n-1} i_{n+1} \ldots i_N \right\rangle \left\langle j_1 \ldots j_{n-1} j_{n+1} \ldots j_N \right| \langle i_n | i_n \rangle \langle j_n | i_n \rangle \nonumber \\ & \quad = \langle j_n | i_n \rangle \left| i_1 \ldots i_{n-1} i_{n+1} \ldots i_N \right\rangle \left\langle j_1 \ldots j_{n-1} j_{n+1} \ldots j_N \right| \nonumber \\ & \quad = \delta_{i_n j_n} \left| i_1 \ldots i_{n-1} i_{n+1} \ldots i_N \right\rangle \left\langle j_1 \ldots j_{n-1} j_{n+1} \ldots j_N \right| \label{NPartialTrace}
\end{align}
such that the $(n-1)$-body reduced density matrix is given by
\begin{align*}
& \rho_{1\, \ldots n-1\, n+1 \ldots N} \\ & \quad = \sum_{\mathcal{I}} \sum_{\mathcal{J}} \psi_{i_1 \ldots i_n \ldots i_N} \bar{\psi}_{j_1 \ldots i_n \ldots j_N} \left| i_1 \ldots i_{n-1} i_{n+1} \ldots i_N \right\rangle \left\langle j_1 \ldots j_{n-1} j_{n+1} \ldots j_N \right|,
\end{align*}
where $\mathcal{I}$ and $\mathcal{J}$ are the index sets.

Permutation matrices can act on the $(n-1)$-body reduced density matrix so that the labeling of qubits can be rearranged. Particularly, we want a cyclic permutation in such a way that qubits labeled 1 to $n-1$ are permuted to the back,
\begin{align*}
& P_\pi \left( \rho_{1\, \ldots n-1\, n+1 \ldots N} \right) P^\text{T}_\pi \nonumber \\ & \, = \rho_{n+1\, \ldots N\, 1\, \ldots n-1} \nonumber \\ & \, = \sum_{\mathcal{I}} \sum_{\mathcal{J}} \psi_{i_1 \ldots i_n \ldots i_N} \bar{\psi}_{j_1 \ldots i_n \ldots j_N} P_\pi \left| i_1 \ldots i_{n-1} i_{n+1} \ldots i_N \right\rangle \left\langle j_1 \ldots j_{n-1} j_{n+1} \ldots j_N \right| P^\text{T}_\pi \nonumber \\ & \, = \sum_{\mathcal{I}} \sum_{\mathcal{J}} \psi_{i_1 \ldots i_n \ldots i_N} \bar{\psi}_{j_1 \ldots i_n \ldots j_N} \left| i_{n+1} \ldots i_N i_1 \ldots i_{n-1} \right\rangle \left\langle j_{n+1} \ldots j_N j_1 \ldots j_{n-1} \right| \nonumber \\ & \, = \Psi_{(n)}^\text{T} \bar{\Psi}_{(n)},
\end{align*}
where $P_\pi$ is the matrix for the desired permutation and $\Psi_{(n)}$ is the $n$-th matrix unfolding of $\Psi$.

In addition, one can perform partial trace operation $N-1$ times on $N$-partite states to obtain a set of one-body reduced density matrices. From equation (\ref{NPartialTrace}), every time an $n$-partial trace operation is performed, a Kronecker delta $\delta_{i_n j_n}$ will be produced. Thus, the one-body reduced density matrices will have the following generic form,
\begin{align}
& \rho_{n} = \sum_{\mathcal{I}} \sum_{\mathcal{J}} \psi_{i_1 \ldots i_n \ldots i_N} \bar{\psi}_{i_1 \ldots j_n \ldots i_N} \left| i_n \right\rangle \left\langle j_n \right| = \Psi_{(n)} \Psi_{(n)}^\dagger. \label{NonebodyRDM}
\end{align}
\end{proof}

\subsection{Theorem \ref{HOSVDReducedDensityMatrixTheorem}} \label{AppendixProofHOSVD}
\setcounter{theorem}{2}
\begin{theorem}[HOSVD and one-body reduced density matrices]
Let $\Psi \in \mathbb{C}^{I_1} \otimes \ldots \otimes \mathbb{C}^{I_n} \otimes \ldots \otimes \mathbb{C}^{I_N}$ be an $N$th-order complex tensor and $\mathcal{T}$ be its core tensor. HOSVD simultaneously diagonalizes the set of one-body reduced density matrices of multipartite states in such a way that the $n$-mode singular values are ordered.
\end{theorem}

\begin{proof}
From equation (\ref{NonebodyRDM}), the summation of the two index sets $\mathcal{I}$ and $\mathcal{J}$ is between two subtensors $\Psi_{i_n}$ and $\Psi_{j_n}$. Due to Theorem \ref{HOSVDMatrixUnfoldingTheorem}, we can write
\begin{align}
\Psi_{(n)} \Psi_{(n)}^\dagger = \rho_n = U^{(n)} T_{(n)} T_{(n)}^\dagger U^{(n)\dagger} = U^{(n)} \rho_n^d U^{(n)\dagger},
\end{align}
where $\rho_n^d = T_{(n)} T_{(n)}^\dagger$ is the $n$-th diagonalized one-body reduced density matrix. The one-body reduced density matrix is diagonalized because when $i_n = j_n$, we obtain the square of $n$-mode singular values, $\sigma_i^{(n)2}$, whereas when $i_n \neq j_n$, we have the all-orthogonality conditions, which are zero due to HOSVD.
\end{proof}

\section{Special three-qubit core tensors by concurrency of three lines} \label{AppendixConcurrencyThreeQubits}

By definition, the all-orthogonality conditions of three qubits are
\begin{align}
\bar{t}_{111} t_{211} + \bar{t}_{121} t_{221} + \bar{t}_{112} t_{212} + \bar{t}_{122} t_{222} & = 0, \\
\bar{t}_{111} t_{121} + \bar{t}_{211} t_{221} + \bar{t}_{112} t_{122} + \bar{t}_{212} t_{222} & = 0, \\
\bar{t}_{111} t_{112} + \bar{t}_{211} t_{212} + \bar{t}_{121} t_{122} + \bar{t}_{221} t_{222} & = 0.
\end{align}
By writing $\bar{t}_{111}$ and $t_{222}$ in terms of other unknowns, we can reformulate the solutions to the above all-orthogonality conditions in the form of concurrency of three lines,
\begin{align}
\begin{vmatrix}
t_{211} & \bar{t}_{122} & \bar{t}_{121} t_{221} + \bar{t}_{112} t_{212} \\
t_{121} & \bar{t}_{212} & \bar{t}_{211} t_{221} + \bar{t}_{112} t_{122} \\
t_{112} & \bar{t}_{221} & \bar{t}_{211} t_{212} + \bar{t}_{121} t_{122}
\end{vmatrix} = 0,
\end{align}
where $(x,\, y) = (\bar{t}_{111}, t_{222})$. Now, we study all possible solutions to the above determinant.

\refstepcounter{subsection}
\subsection*{\thesubsection \quad Column 1 = 0: $t_{112} = t_{121} = t_{211} = 0$}

We have
\begin{align}
\bar{t}_{122} t_{222} & = 0, \\
\bar{t}_{212} t_{222} & = 0, \\
\bar{t}_{221} t_{222} & = 0.
\end{align}
Let $t_{222} = 0$, we have $\left| \text{B}_1 \right\rangle = t_{111} \left| 111 \right\rangle + t_{122} \left| 122 \right\rangle + t_{212} \left| 212 \right\rangle + t_{221} \left| 221 \right\rangle$.

\refstepcounter{subsection}
\subsection*{\thesubsection \quad Column 2 = 0: $t_{122} = t_{212} = t_{221} = 0$}

We have
\begin{align}
\bar{t}_{111} t_{211} & = 0, \\
\bar{t}_{111} t_{121} & = 0, \\
\bar{t}_{111} t_{112} & = 0.
\end{align}
Let $t_{111} = 0$, we have $\left| \text{B}_2 \right\rangle = t_{112} \left| 112 \right\rangle + t_{121} \left| 121 \right\rangle + t_{211} \left| 211 \right\rangle + t_{222} \left| 222 \right\rangle$.

\refstepcounter{subsection}
\subsection*{\thesubsection \quad Column 3 = 0}

We have
\begin{align}
\bar{t}_{121} t_{221} + \bar{t}_{112} t_{212} & = 0, \label{Column3|1} \\
\bar{t}_{211} t_{221} + \bar{t}_{112} t_{122} & = 0, \label{Column3|2} \\
\bar{t}_{211} t_{212} + \bar{t}_{121} t_{122} & = 0, \label{Column3|3}
\end{align}
in addition to the original all-orthogonality conditions that have to be satisfied, which are reduced to
\begin{align}
\bar{t}_{111} t_{211} + \bar{t}_{122} t_{222} & = 0, \label{Column3|4} \\
\bar{t}_{111} t_{121} + \bar{t}_{212} t_{222} & = 0, \label{Column3|5} \\
\bar{t}_{111} t_{112} + \bar{t}_{221} t_{222} & = 0. \label{Column3|6}
\end{align}

From equations (\ref{Column3|4}) and (\ref{Column3|6}), since we are looking for minimum requirements to satisfy the set of equations, we have to let $t_{111} = t_{222} = 0$. From equations (\ref{Column3|1}) and (\ref{Column3|2}), we have
\begin{align}
\frac{\bar{t}_{112}}{t_{221}} = - \frac{\bar{t}_{121}}{t_{212}} = - \frac{\bar{t}_{211}}{t_{122}},
\end{align}
but
\begin{align}
\frac{\bar{t}_{121}}{t_{212}} = - \frac{\bar{t}_{211}}{t_{122}}
\end{align}
from equation (\ref{Column3|3}). In order to resolve this contradiction, we consider the following possibilities:-
\begin{enumerate}
\item $t_{112} = t_{221} = 0$
\newline
$\left| \psi \right\rangle = t_{121} \left| 121 \right\rangle + t_{122} \left| 122 \right\rangle + t_{211} \left| 211 \right\rangle + t_{212} \left| 212 \right\rangle,\, \bar{t}_{211} t_{212} + \bar{t}_{121} t_{122} = 0$;
\newline
$\sigma_1^{(1)2} = \sigma_1^{(3)2} = \left| t_{121} \right|^2 + \left| t_{122} \right|^2,\, \sigma_1^{(2)2} = \sigma_2^{(1)2} = \left| t_{211} \right|^2 + \left| t_{212} \right|^2$.
\newline
From the ordering property of higher order singular value decomposition, since $\sigma_1^{(2)2}$ is the largest 2-mode singular value, the only way $\sigma_1^{(2)2} = \sigma_2^{(1)2}$ can be satisfied is when $\sigma_1^{(2)2} = \frac{1}{2}$, resulting to $\sigma_1^{(1)2} = \sigma_1^{(2)2} = \sigma_1^{(3)2} = \frac{1}{2}$. This is not a generic special state that we are looking for.

\item $t_{121} = t_{212} = 0$
\newline
$\left| \psi \right\rangle = t_{112} \left| 112 \right\rangle + t_{122} \left| 122 \right\rangle + t_{211} \left| 211 \right\rangle + t_{221} \left| 221 \right\rangle,\, \bar{t}_{211} t_{221} + \bar{t}_{112} t_{122} = 0$;
\newline
$\sigma_1^{(1)2} = \sigma_1^{(2)2} = \left| t_{112} \right|^2 + \left| t_{122} \right|^2,\, \sigma_1^{(3)2} = \sigma_2^{(1)2} = \left| t_{211} \right|^2 + \left| t_{221} \right|^2$.
\newline
From the ordering property of higher order singular value decomposition, since $\sigma_1^{(3)2}$ is the largest 3-mode singular value, the only way $\sigma_1^{(3)2} = \sigma_2^{(1)2}$ can be satisfied is when $\sigma_1^{(3)2} = \frac{1}{2}$, resulting to $\sigma_1^{(1)2} = \sigma_1^{(2)2} = \sigma_1^{(3)2} = \frac{1}{2}$. This is not a generic special state that we are looking for.

\item $t_{122} = t_{211} = 0$
\newline
$\left| \psi \right\rangle = t_{112} \left| 112 \right\rangle + t_{121} \left| 121 \right\rangle + t_{212} \left| 212 \right\rangle + t_{221} \left| 221 \right\rangle,\, \bar{t}_{121} t_{221} + \bar{t}_{112} t_{212} = 0$;
\newline
$\sigma_1^{(1)2} = \sigma_1^{(2)2} = \left| t_{112} \right|^2 + \left| t_{121} \right|^2,\, \sigma_1^{(3)2} = \sigma_2^{(2)2} = \left| t_{121} \right|^2 + \left| t_{221} \right|^2$.
\newline
From the ordering property of higher order singular value decomposition, since $\sigma_1^{(3)2}$ is the largest 3-mode singular value, the only way $\sigma_1^{(3)2} = \sigma_2^{(2)2}$ can be satisfied is when $\sigma_1^{(3)2} = \frac{1}{2}$, resulting to $\sigma_1^{(1)2} = \sigma_1^{(2)2} = \sigma_1^{(3)2} = \frac{1}{2}$. This is not a generic special state that we are looking for.
\end{enumerate}

\refstepcounter{subsection}
\subsection*{\thesubsection \quad Row 1 = 0}

We have
\begin{align}
\bar{t}_{112} t_{212} + \bar{t}_{121} t_{221} & = 0, \label{Row1|1} \\
\bar{t}_{111} t_{121} + \bar{t}_{212} t_{222} & = 0, \label{Row1|2} \\
\bar{t}_{111} t_{112} + \bar{t}_{221} t_{222} & = 0. \label{Row1|3}
\end{align}

From equations (\ref{Row1|2}) and (\ref{Row1|3}), we have
\begin{align}
\frac{\bar{t}_{111}}{t_{222}} = - \frac{\bar{t}_{212}}{t_{121}} = - \frac{\bar{t}_{221}}{t_{112}},
\end{align}
but
\begin{align}
\frac{\bar{t}_{212}}{t_{121}} = - \frac{\bar{t}_{221}}{t_{112}}
\end{align}
from equation (\ref{Row1|1}). In order to resolve this contradiction, we consider the following possibilities:-
\begin{enumerate}
\item $t_{111} = t_{222} = 0$
\newline
$\left| \psi \right\rangle = t_{112} \left| 112 \right\rangle + t_{121} \left| 121 \right\rangle + t_{212} \left| 212 \right\rangle + t_{221} \left| 221 \right\rangle,\, \bar{t}_{121} t_{221} + \bar{t}_{112} t_{212} = 0$;
\newline
$\sigma_1^{(1)2} = \sigma_1^{(2)2} = \left| t_{112} \right|^2 + \left| t_{121} \right|^2,\, \sigma_1^{(3)2} = \sigma_2^{(2)2} = \left| t_{121} \right|^2 + \left| t_{221} \right|^2$.
\newline
From the ordering property of higher order singular value decomposition, since $\sigma_1^{(3)2}$ is the largest 3-mode singular value, the only way $\sigma_1^{(3)2} = \sigma_2^{(2)2}$ can be satisfied is when $\sigma_1^{(3)2} = \frac{1}{2}$, resulting to $\sigma_1^{(1)2} = \sigma_1^{(2)2} = \sigma_1^{(3)2} = \frac{1}{2}$. This is not a generic special state that we are looking for.

\item $t_{112} = t_{221} = 0$
\newline
$\left| \text{S}_2 \right\rangle = t_{111} \left| 111 \right\rangle + t_{121} \left| 121 \right\rangle + t_{212} \left| 212 \right\rangle + t_{222} \left| 222 \right\rangle,\, \bar{t}_{111} t_{121} + \bar{t}_{212} t_{222} = 0$.

\item $t_{121} = t_{212} = 0$
\newline
$\left| \text{S}_1 \right\rangle = t_{111} \left| 111 \right\rangle + t_{112} \left| 112 \right\rangle + t_{221} \left| 221 \right\rangle + t_{222} \left| 222 \right\rangle,\, \bar{t}_{111} t_{112} + \bar{t}_{221} t_{222} = 0$.
\end{enumerate}

\refstepcounter{subsection}
\subsection*{\thesubsection \quad Row 2 = 0}

We have
\begin{align}
\bar{t}_{211} t_{221} + \bar{t}_{112} t_{122} & = 0, \label{Row2|1} \\
\bar{t}_{111} t_{211} + \bar{t}_{122} t_{222} & = 0, \label{Row2|2} \\
\bar{t}_{111} t_{112} + \bar{t}_{221} t_{222} & = 0. \label{Row2|3}
\end{align}

From equations (\ref{Row2|2}) and (\ref{Row2|3}), we have
\begin{align}
\frac{\bar{t}_{111}}{t_{222}} = - \frac{\bar{t}_{122}}{t_{211}} = - \frac{\bar{t}_{221}}{t_{112}},
\end{align}
but
\begin{align}
\frac{\bar{t}_{122}}{t_{211}} = - \frac{\bar{t}_{221}}{t_{112}}
\end{align}
from equation (\ref{Row2|1}). In order to resolve this contradiction, we consider the following possibilities:-
\begin{enumerate}
\item $t_{111} = t_{222} = 0$
\newline
$\left| \psi \right\rangle = t_{121} \left| 121 \right\rangle + t_{122} \left| 122 \right\rangle + t_{211} \left| 211 \right\rangle + t_{212} \left| 212 \right\rangle,\, \bar{t}_{211} t_{212} + \bar{t}_{121} t_{122} = 0$;
\newline
$\sigma_1^{(1)2} = \sigma_1^{(3)2} = \left| t_{121} \right|^2 + \left| t_{122} \right|^2,\, \sigma_1^{(2)2} = \sigma_2^{(1)2} = \left| t_{211} \right|^2 + \left| t_{212} \right|^2$.
\newline
From the ordering property of higher order singular value decomposition, since $\sigma_1^{(2)2}$ is the largest 2-mode singular value, the only way $\sigma_1^{(2)2} = \sigma_2^{(1)2}$ can be satisfied is when $\sigma_1^{(2)2} = \frac{1}{2}$, resulting to $\sigma_1^{(1)2} = \sigma_1^{(2)2} = \sigma_1^{(3)2} = \frac{1}{2}$. This is not a generic special state that we are looking for.

\item $t_{122} = t_{211} = 0$
\newline
$\left| \text{S}_3 \right\rangle = t_{111} \left| 111 \right\rangle + t_{122} \left| 122 \right\rangle + t_{211} \left| 211 \right\rangle + t_{222} \left| 222 \right\rangle,\, \bar{t}_{111} t_{211} + \bar{t}_{122} t_{222} = 0$.

\item $t_{122} = t_{211} = 0$
\newline
$\left| \text{S}_1 \right\rangle = t_{111} \left| 111 \right\rangle + t_{112} \left| 112 \right\rangle + t_{221} \left| 221 \right\rangle + t_{222} \left| 222 \right\rangle,\, \bar{t}_{111} t_{112} + \bar{t}_{221} t_{222} = 0$.
\end{enumerate}

\refstepcounter{subsection}
\subsection*{\thesubsection \quad Row 3 = 0}

We have
\begin{align}
\bar{t}_{211} t_{212} + \bar{t}_{121} t_{122} & = 0, \label{Row3|1} \\
\bar{t}_{111} t_{211} + \bar{t}_{122} t_{222} & = 0, \label{Row3|2} \\
\bar{t}_{111} t_{121} + \bar{t}_{212} t_{222} & = 0. \label{Row3|3}
\end{align}

From equations (\ref{Row3|2}) and (\ref{Row3|3}), we have
\begin{align}
\frac{\bar{t}_{111}}{t_{222}} = - \frac{\bar{t}_{122}}{t_{211}} = - \frac{\bar{t}_{212}}{t_{121}},
\end{align}
but
\begin{align}
\frac{\bar{t}_{122}}{t_{211}} = - \frac{\bar{t}_{212}}{t_{121}}
\end{align}
from equation (\ref{Row3|1}). In order to resolve this contradiction, we consider the following possibilities:-
\begin{enumerate}
\item $t_{111} = t_{222} = 0$
\newline
$\left| \psi \right\rangle = t_{112} \left| 112 \right\rangle + t_{122} \left| 122 \right\rangle + t_{211} \left| 211 \right\rangle + t_{221} \left| 221 \right\rangle,\, \bar{t}_{211} t_{221} + \bar{t}_{112} t_{122} = 0$;
\newline
$\sigma_1^{(1)2} = \sigma_1^{(2)2} = \left| t_{112} \right|^2 + \left| t_{122} \right|^2,\, \sigma_1^{(3)2} = \sigma_2^{(1)2} = \left| t_{211} \right|^2 + \left| t_{221} \right|^2$.
\newline
From the ordering property of higher order singular value decomposition, since $\sigma_1^{(3)2}$ is the largest 3-mode singular value, the only way $\sigma_1^{(3)2} = \sigma_2^{(1)2}$ can be satisfied is when $\sigma_1^{(3)2} = \frac{1}{2}$, resulting to $\sigma_1^{(1)2} = \sigma_1^{(2)2} = \sigma_1^{(3)2} = \frac{1}{2}$. This is not a generic special state that we are looking for.

\item $t_{122} = t_{211} = 0$
\newline
$\left| \text{S}_2 \right\rangle = t_{111} \left| 111 \right\rangle + t_{121} \left| 121 \right\rangle + t_{212} \left| 212 \right\rangle + t_{222} \left| 222 \right\rangle,\, \bar{t}_{111} t_{121} + \bar{t}_{212} t_{222} = 0$.

\item $t_{121} = t_{212} = 0$
\newline
$\left| \text{S}_3 \right\rangle = t_{111} \left| 111 \right\rangle + t_{122} \left| 122 \right\rangle + t_{211} \left| 211 \right\rangle + t_{222} \left| 222 \right\rangle,\, \bar{t}_{111} t_{211} + \bar{t}_{122} t_{222} = 0$.
\end{enumerate}

\end{document}